\documentclass{article}
\usepackage[utf8]{inputenc}
\usepackage[margin=1in]{geometry}
\usepackage{macros}

\title{Improved Low-Depth Set-Multilinear Circuit Lower Bounds}
\author{Deepanshu Kush \\University of Toronto \\
\href{mailto:deepkush@cs.toronto.edu}{\texttt{deepkush@cs.toronto.edu}}\and Shubhangi Saraf \\University of Toronto\\ \href{mailto:shubhangi.saraf@utoronto.ca}{\texttt{shubhangi.saraf@utoronto.ca}}}

\begin{document}
\maketitle

\begin{abstract}
In this paper, we prove strengthened lower bounds for constant-depth
set-multilinear formulas. More precisely, we show that over any field, there is an explicit polynomial $f$ in VNP defined over $n^2$ variables, and of degree $n$, such that any product-depth $\Delta$ set-multilinear formula computing $f$ has size at least $n^{\Omega \left( n^{1/\Delta}/\Delta\right)  }$.
The hard polynomial $f$ comes from the class of Nisan-Wigderson (NW) design-based polynomials. 

Our lower bounds improve upon the recent work of Limaye, Srinivasan and Tavenas (STOC 2022), where a lower bound of the form $(\log n)^{\Omega (\Delta n^{1/\Delta})}$ was shown for the size of product-depth $\Delta$ set-multilinear formulas computing the iterated matrix multiplication (IMM) polynomial of the same degree and over the same number of variables as $f$. Moreover, our lower bounds are novel for any $\Delta\geq 2$.

The precise quantitative expression in our lower bound is interesting also because the lower bounds we obtain are ``sharp'' in the sense that any asymptotic improvement would imply general set-multilinear circuit lower bounds via depth reduction results. 

In the setting of general set-multilinear formulas, a lower bound of the form $ n^{\Omega(\log n)}$ was already obtained by Raz (J. ACM 2009) for the more general model of multilinear formulas. The techniques of LST (which extend the techniques of the same authors in (FOCS 2021)) give a different route to set-multilinear formula lower bounds, and allow them to obtain  a lower bound of the form $(\log n)^{\Omega(\log n)}$ for the size of general set-multilinear formulas computing the IMM polynomial. Our proof techniques are another variation on those of LST, and enable us to show an improved lower bound (matching that of Raz) of the form $n^{\Omega(\log n)}$, albeit for the same polynomial $f$ in VNP (the NW polynomial).
As observed by LST, if the same $n^{\Omega(\log n)}$ size lower bounds for unbounded-depth set-multilinear formulas could be obtained for the IMM polynomial, then using the self-reducibility of IMM and using hardness escalation results, this would imply super-polynomial lower bounds for general algebraic formulas. 

\end{abstract}

\newpage

\section{Introduction}\label{sec:intro}

\paragraph{Background.}
An \emph{algebraic circuit} over a field $\F$ for a multivariate polynomial $P(x_1,\ldots,x_N)$ is a directed acyclic graph (DAG)
whose internal vertices (called gates) are labeled as either $+$ (sum) or $\times$ (product), and leaves
(vertices of in-degree zero) are labeled by the variables $x_i$ or constants from $\F$. A special output gate (the root of the DAG) represents the polynomial
$P$. If the DAG happens to be a tree, such a resulting circuit is called an \emph{algebraic formula}. 
The size of a circuit is the number of nodes in the DAG. We also consider the product-depth
of the circuit, which is the maximum number of product gates on a root-to-leaf path.

An algebraic circuit is therefore a computational model, which solves the computational
task of evaluating $P$ on a given input $(x_1,\ldots,x_N)$. The complexity of this model is measured
by the size of the circuit, which serves as an indicator of the
time complexity of computing the polynomial. The product-depth measures the degree to which this computation can be made parallel. 
As an algebraic circuit is supposed to construct a formal polynomial $P$, it is a \emph{syntactic} model of
computation. This is unlike a Boolean circuit, which is only required to model specific
truth-table constraints. The problem of proving algebraic circuit lower
bounds is therefore widely considered to be easier than its Boolean counterpart. Indeed, we know that
proving VP $\neq$ VNP, the algebraic analog of the P vs. NP problem, is implied by the latter separation, in the non-uniform setting (\cite{burg}). We refer the reader to \cite{Saptarishi-survey} for a much more elaborate survey of this topic.

\paragraph{The LST breakthrough.}
Much like in the Boolean setting, the problem of showing lower bounds for \emph{general} algebraic circuits (or even formulas) has remained elusive. However, some remarkable progress has been made very recently by Limaye, Srinivasan, and Tavenas (\cite{LST1}) who in a spectacular breakthrough, showed the first super-polynomial lower bounds for algebraic circuits of \emph{all}
constant depths. Prior to their work, the best known lower bound (\cite{KayalST16}) even for product-depth 1 (or $\Sigma\Pi\Sigma$ circuits) was only almost-cubic. This is in stark contrast with the Boolean setting, in which we have known strong constant-depth lower bounds for many decades \cite{Ajtai83, FurstSS84,Yao85,Hastad86,Razborov1987LowerBO,Smolensky87}. Constant-depth circuits are critical to the study of algebraic complexity theory, as unlike the Boolean setting, strong enough bounds against them are known to yield VP $\neq$ VNP (\cite{AgrawalV08}). This helps put into perspective the importance of the work \cite{LST1}.

The crucial step in the proof of their result is to first establish super-polynomial lower bounds for a certain restricted class of (low-depth) algebraic circuits, namely \emph{set-multilinear} circuits which we now define along with other important circuit models. A polynomial is said to be homogeneous
if each monomial has the same total degree and \emph{multilinear} if every variable occurs
at most once in any monomial. Now, suppose that the underlying variable set is partitioned into $d$ sets
$X_1,\ldots, X_d$. Then the polynomial is said to be \emph{set-multilinear} with respect to this variable partition if each
monomial in $P$ has \emph{exactly} one variable from each set. We also define different models
of computation corresponding to these variants of polynomials classes. An algebraic formula (circuit) is set-multilinear with respect to a variable partition
$(X_1,\ldots, X_d)$ if each internal node in the formula (circuit) computes a set-multilinear polynomial. Multilinear/homogeneous circuits and formulas are defined analogously.  

Several well-studied and interesting polynomials happen to be set-multilinear. For example, the Determinant
and the Permanent polynomials, the study of which is profoundly consequential to the field of algebraic complexity theory, are
set-multilinear (with respect to the column variables). Another well-studied polynomial, namely the Iterated
Matrix Multiplication polynomial, is also set-multilinear. The polynomial IMM$_{n,d}$ is defined on $N= dn^2$ variables, where the variables are partitioned into $d$ sets $X_1,\ldots, X_d$ of
size $n^2$, each of which is represented as an $n\times n$ matrix with distinct variable entries. The
polynomial IMM$_{n,d}$ is defined to be the polynomial that is the $(1, 1)$-th entry of the product
matrix $X_1 X_2\cdots X_d$. This polynomial has a simple divide-and-conquer-based set-multilinear formula of size
$n^{O(\log d)}$, and more generally for every $\Delta\leq \log d$, a set-multilinear formula of product-depth $\Delta$ and size $n^{O(\Delta d^{1/\Delta})}$, and circuit\footnote{In this paper, when speaking of constant-depth models of computation at a high level, we shall often use the terms circuit and formula interchangeably as a product-depth $\Delta$ circuit of size $s$ can be simulated by a product-depth $\Delta$ formula of size $s^{2\Delta}$.} of size $n^{O( d^{1/\Delta})}$. Even without the set-multilinearity constraint, no significantly better upper
bound is known. It is reasonable to conjecture that
this simple upper bound is tight up to the constant in the exponent. 

The lower bounds in \cite{LST1}
for general constant-depth algebraic circuits are shown in the following sequence of steps:
\begin{enumerate}
    \item It is shown that general low-depth algebraic circuits can be transformed to set-multilinear algebraic circuits of low depth, and without much of a blow-up in size (as long as the degree is small). More precisely, any product-depth $\Delta$ circuit of size $s$ computing a polynomial that is set-multilinear with respect to the partition $(X_1,\ldots,X_d)$ where each $|X_i|\leq n$, can be converted to a set-multilinear circuit\footnote{There is also an intermediate `homogenization' step which we skip describing here for the sake of brevity.} of product-depth $2\Delta$ and size $\poly(s)\cdot d^{O(d)}$. Such a `set-multilinearization' of general formulas of small degree was already shown before in \cite{Raz-Tensor} (which we describe soon in more detail); however, the main contribution of \cite{LST1} here is to prove this \emph{depth-preserving} version of it.
    \item Strong lower bounds are then established for low-depth set-multilinear circuits (of small enough degree). More precisely, any set-multilinear circuit $C$ computing IMM$_{n,d}$ (where $d = O(\log n)$) of product-depth $\Delta$ must have size at least $n^{d^{\exp(-O(\Delta))}}$. This combined with the first step yields the desired lower bound for general constant-depth circuits.
\end{enumerate}

Given Raz's set-multilinearization of formulas of small degree that we alluded to, and this description of the set-multilinear formula lower bounds from \cite{LST1} when $d = O(\log n)$, it is evident the `small degree' regime is inherently interesting to study - as it provides an avenue, via `hardness escalation', for tackling one of the grand challenges of algebraic complexity theory, namely proving super-polynomial lower bounds for general algebraic formulas. However, we shall now see that even the large degree regime can be equally consequential in this regard.

\paragraph{The large degree regime.}
Consider a polynomial $P$ that is set-multilinear with respect to the variable partition $(X_1,\ldots,X_d)$ where each $|X_i|\leq n$. The main focus of this paper is to study set-multilinear circuit complexity in the regime where $d$ and $n$ are \emph{polynomially} related (as opposed to say, the assumption $d = O(\log n)$ described above). We now provide some background and motivation for studying this regime.

In follow-up work \cite{LST2}, the same authors showed the first super-polynomial
lower bound against unbounded-depth set-multilinear formulas computing IMM$_{n,n}$\footnote{Note that for IMM$_{n,n}$, each $X_i$ has size $n^2$, not $n$. But the important thing for us here is that the degree, $n$, is polynomially related to this parameter.}. As is astutely described in \cite{LST2}, studying the set-multilinear formula complexity of IMM is extremely interesting and consequential even in the setting $d=n$ because of the following reasons:
\begin{itemize}
\item IMM$_{n,n}$ is a \emph{self-reducible} polynomial i.e., it is  possible to construct formulas for IMM$_{n,n}$ by recursively using formulas for IMM$_{n,d}$ (for any
$d<n$). In particular, if we had formulas of size $n^{o(\log d)}$ for IMM$_{n,d}$ (for some $d<n$), this would imply formulas
of size $n^{o(\log n)}$ for IMM$_{n,n}$. In other words, an optimal $n^{\Omega(\log n)}$
lower bound for IMM$_{n,n}$ implies $n^{\omega_d(1)}$ lower bounds for IMM$_{n,d}$ for any $d< n$.
\item Raz in \cite{Raz-Tensor} showed that if an $N$-variate set-multilinear
polynomial of degree $d$ has an algebraic formula of size $s$, then it also has a set-multilinear
formula of size $\poly(s)\cdot (\log s)^{d}$. In particular, for a set-multilinear polynomial $P$ of degree
$d = O(\log N/ \log \log N)$, it follows that $P$ has a formula of size $\poly(N)$ if and only if $P$ has a
set-multilinear formula of size $\poly(N)$. Thus, having $N^{\omega_d(1)}$ set-multilinear
formula size lower bounds for such a low degree would imply super-polynomial lower bounds for general formulas. 
\end{itemize}

In particular, proving the optimal $n^{\Omega(\log n)}$ set-multilinear formula size lower bound for IMM$_{n,n}$ would have dramatic consequences. To this end, the authors in \cite{LST2} are able to show a weaker bound of the form $(\log n)^{\Omega(\log n)}$ instead. Even though it is the case that `simply' improving the base of this exponent from $\log n$ to $n$ yields general formula lower bounds, it seems that we are still far from achieving it. Indeed, as is observed in \cite{LST2}, we do not even have the optimal $n^{\Omega(\sqrt{n})}$ lower bound\footnote{This is known for set-multilinear (and even multilinear) $\Sigma\Pi\Sigma\Pi$ circuits (see \cite{FournierLMS15,Kayal0T18}), but those are only special cases of general product-depth $2$ circuits, which are $\Sigma\Pi\Sigma\Pi\Sigma$.} when product-depth $\Delta = 2$.
Moreover, we do not know how to obtain a lower bound of the form $n^{\Omega(\sqrt{n})}$ for product-depth $2$ set-multilinear circuits for \emph{any} explicit polynomial of degree $n$ and in $\poly(n)$ variables. For product-depths $\Delta\leq \log n$, \cite{LST2} shows a set-multilinear formula size lower bound of $(\log n)^{\Omega(\Delta n^{1/\Delta})}$ for IMM$_{n,n}$, which is in fact the best set-multilinear lower bound we know for any polynomial of degree $n$ and in $\poly(n)$ variables, and for any $\Delta \geq 2$. As far as we know, the previous best lower bound of $\exp(\Omega(n^{1/\Delta}))$, also for IMM$_{n,n}$, followed from the work of Nisan and Wigderson (\cite{NisanW97}). It is therefore an interesting challenge to improve the base of this exponent from $\log n$ to $n$ i.e., establish a near-optimal $n^{\Omega( n^{1/\Delta})}$ lower bound in the constant (or low) depth setting.

\paragraph{Our Results.}
In this paper, we obtain these ``optimal'' lower bounds, albeit not for IMM$_{n,n}$, but rather for another explicit polynomial in VNP. We show the following:

\begin{theorem}\label{thm-intro:main-bd-depth}
Let $N$ be a growing parameter and $\Delta$ be an integer such that $1\leq \Delta \leq \log N/\log \log N$. There is an explicit polynomial $P_N$ defined over $N = n^2$ variables with degree $d = n$ that is set-multilinear with respect to the variable partition $X = (X_1,\ldots, X_d)$ where each $|X_i| = n$ and such that any set-multilinear
formula of product-depth $\Delta$ computing $P_N(X)$ must have size at least $N^{\Omega(d^{1/\Delta}/\Delta)}$.
\end{theorem}

Notice that obtaining this precise bound is interesting also when viewed through the lens of \emph{depth reduction}. Tavenas (\cite{Tavenas15}), building on several prior works (\cite{AgrawalV08, Koiran12}), showed that any algebraic circuit of $\poly(N)$ size computing a homogeneous $N$-variate polynomial of degree $d$  can be converted to a homogeneous circuit of product-depth\footnote{The result is stated in \cite{Tavenas15} for $\Sigma\Pi\Sigma\Pi$ circuits but the proof can be appropriately modified for larger product-depths.} $\Delta$ of size $(Nd)^{O(d^{1/\Delta})}$. It easily follows from the proof that this depth
reduction preserves syntactic restrictions. That is, if we start with a syntactically set-multilinear circuit, the resulting product-depth $\Delta$ circuit is also syntactically set-multilinear. Therefore, the precise bound in Theorem \ref{thm-intro:main-bd-depth} is \emph{sharp} in the sense that any asymptotic improvement in its exponent would imply super-polynomial set-multilinear circuit lower bounds, which would be quite a strong and interesting consequence.
Another very intriguing direction is to consider the problem of \emph{improved} depth reduction for set-multilinear circuits. If an asymptotic improvement in the exponent on the bound for general circuits from \cite{Tavenas15} could be shown to hold for set-multilinear circuits in the setting of Theorem \ref{thm-intro:main-bd-depth} (i.e., when $N= d^2$), this would again imply super-polynomial set-multilinear circuit lower bounds. There is some evidence towards this possibility, as \cite{KOS19} shows such an improvement in a certain regime of parameters for multilinear circuits (see the discussion in Section \ref{sec:open} for more details). 

\begin{remark}\label{rem:true-bd}
The lower bound in Theorem \ref{thm-intro:main-bd-depth} is actually $d^{\Omega(d^{1/\Delta}/\Delta)}$, where $d$ is the degree of the underlying polynomial, and it holds as long as degree $d\leq n$ (the details are deferred to the proof of Theorem \ref{thm:main-bd-depth} in Section \ref{sec:main}). Observe that for constant $\Delta$ this bound already nearly matches the bound $(\log n)^{\Omega(\Delta d^{1/\Delta})}$ in \cite{LST2} (which was shown for IMM$_{n,d}$) when $d = (\log n)^{\Omega(1)}$ and exceeds it as soon as $d$ becomes super-polylogarithmic in $n$. Moreover for $d < \log n/ \log \log n$, both the bounds are trivial even for $\Delta =1$.
\end{remark}

We also remark that in several lower bounds for algebraic circuit classes in the past, the lower bound was initially shown for a polynomial in VNP and then with additional effort, was shown to also hold for a polynomial in VP (in particular, the IMM polynomial). A strong candidate for the choice of this polynomial family in VNP has been the Nisan-Wigderson (NW) design-based (\cite{NisanW94-hardness}) family of polynomials. For instance, \cite{KayalSS14} showed a lower bound of $n^{\Omega(\sqrt{n})}$ for the top fan-in of a $\Sigma\Pi^{[O(\sqrt{n})]}\Sigma\Pi^{[\sqrt{n}]}$ circuit computing the NW polynomial, which was subsequently shown for IMM by \cite{FournierLMS15}. Similarly, \cite{KayalLSS17} showed an $n^{\Omega(\sqrt{d})}$ size lower bound for homogeneous depth-4 algebraic
formulas for the NW polynomial, which was then shown for IMM later in \cite{/KumarS17}.
Much like these examples, our hard polynomial family in Theorem \ref{thm-intro:main-bd-depth} is also indeed the NW polynomial family, as we shall see in Section \ref{sec:main}.  Our motivation to study constant-depth set-multilinear formula complexity was to prove the optimal lower bounds for the IMM polynomial. Although we are presently able to show it only for the NW polynomial instead of IMM, we are hopeful that this is an important step in its direction. 

In addition to our lower bound for bounded-depth set-multilinear formulas, we observe that the same proof technique also implies a lower bound of the form $n^{\Omega(\log n)}$ for unbounded-depth set-multilinear formulas. \cite{LST2} showed a weaker bound of the form $(\log n)^{\Omega(\log n)}$ but for IMM$_{n,n}$. 

\begin{theorem}\label{thm-intro:main-gen-depth}
For a given integer $N$, there is an explicit polynomial $P_N$ defined over $N = n^2$ variables with degree $d = n$ that is set-multilinear with respect to the variable partition $X = (X_1,\ldots, X_d)$ where each $|X_i| = n$ such that any set-multilinear
formula computing  $P_N(X)$ must have size at least $N^{\Omega(\log N)}$.
\end{theorem}

The hard polynomial in Theorem \ref{thm-intro:main-gen-depth} is also the NW polynomial, which if `improved' to IMM$_{n,n}$, then as discussed, would yield super-polynomial general formula lower bounds. However, we note that in this case, our result is in some sense subsumed by the result of Raz (\cite{Raz09}) who showed an $n^{\Omega(\log n)}$ lower bound for the $n\times n$ permanent (or determinant) polynomial for unbounded-depth multilinear formulas.

\paragraph{Other Related Work.}

In the bounded-depth setting, other than the works \cite{LST1,LST2,NisanW97} already mentioned, there have been several lower bounds for the class of low-depth \emph{multilinear} circuits (\cite{RazY09,ChillaraL019,ChillaraEL018,KayalNS20}). In the unbounded-depth setting, apart from the works \cite{LST2,Raz09} already mentioned for set-multilinear formulas, there have also been several strong lower bounds of the form $n^{\Omega(\log n)}$ against \emph{multilinear} formulas (\cite{DvirMPY12,HrubesY11,Kayal0T18}). However, in both settings of depth, several of these works are not even applicable to the set-multilinear setting as the corresponding hard polynomial does not happen to be set-multilinear.

\paragraph{Proof overview.}

Our overall proof techniques are similar to that of many known lower bounds. We work with a measure that we show to be small for all polynomials computed by small enough set-multilinear
formulas (appropriately so in the bounded and unbounded-depth settings)  and large for the NW polynomial. These \emph{partial derivative measures} were introduced by
Nisan and Wigderson in \cite{NisanW97}, who used them to prove the constant-depth set-multilinear formula lower bounds we discussed earlier. \cite{LST1,LST2} use a particular variant of this measure and our measure is in turn inspired from these works. 

Given a variable partition $(X_1,\ldots, X_d)$, we label each set of variables $X_i$ as `positive' or `negative' uniformly at random. Let $\calP$ and $\calN$ denote the set of positive and negative indices respectively, and let $\calM^\calP$ and $\calM^\calN$ denote the sets of all set-multilinear monomials over $\calP$ and $\calN$ respectively. For a polynomial that is set-multilinear over the given variable partition $(X_1,\ldots, X_d)$, our measure then is simply the rank of the `partial derivative matrix' whose rows are indexed by the elements
of $\calM^{\calP}$ and columns indexed by $\calN^{\calP}$, and the entry of this matrix corresponding to a
row $m_1$ and a column $m_2$ is the coefficient of the monomial $m_1\cdot m_2$ in the given polynomial.

In contrast, the measure used in \cite{LST1} is deterministic and moreover, it is \emph{asymmetric} with respect to the positive and negative variable sets, in the sense that while keeping the positive variable sets as is, it first reduces the size of the negative variable sets by arbitrarily setting a few of these variables to field constants, and then works with the resulting polynomial. On the other hand, \cite{LST2} does use a randomized measure, but one that is still asymmetric, relying on randomly setting a few of the variables inside each set to constants. The way they control the discrepancy between the sizes of the positive and negative variable sets (which is indeed crucial for obtaining the claimed lower bounds) is by imposing a Martingale-like distribution. The lower bound of \cite{NisanW97} also uses random restrictions to enable them to effectively ``simplify" the circuit and upper bound its complexity. Our symmetric, randomized measure avoids random restrictions altogether, and though it is inspired by the measure and the techniques from~\cite{LST1}, it is also reminiscent of the measures used in \cite{Raz09,RazY09} to prove multilinear formula lower bounds.

\section{Preliminaries}

We begin by defining the hard polynomial of our main result (Theorem \ref{thm-intro:main-bd-depth}). As is done in previous lower bounds using the NW polynomials (for example, see \cite{KayalSS14}), we will identify the set of the first $n$ integers as elements of $\F_n$ via an arbitrary correspondence $\phi : [n] \rightarrow \F_n$. If $f(z) \in \F_n[z]$ is a univariate polynomial, then we abuse notation to let $f(i)$ denote the evaluation
of $f$ at the $i$-th field element via the above correspondence i.e., $f(i)\coloneqq \phi^{-1}
(f(\phi(i)))$. To simplify the exposition, in the following definition, we will omit the correspondence $\phi$ and identify a variable
$x_{i,j}$ by the point $(\phi(i), \phi(j)) \in \F_n \times \F_n$.

\begin{definition}[Nisan-Wigderson Polynomials]\label{def:NW}
For a prime power $n$,
let $\F_n$ be a field of size $n$. For an integer $d\leq n$ and the set $X$ of $nd$ variables
$\{x_{i,j} : i\in[n], j \in [d]\}$, we define the degree $d$
homogeneous polynomial $NW_{n,d}$ over any field as
\[
NW_{n,d}(X) = \sum_{\substack{f(z)\in\F_n[z]
\\ \deg(f)<d/2}} \prod_{j\in[d]} x_{f(j),j}.
\]
\end{definition}

Next, we turn to the measure that we shall use to prove Theorems \ref{thm-intro:main-bd-depth} and \ref{thm-intro:main-gen-depth}. For the purpose of setting it up, we follow the notation of \cite{LST1} in the following definition. However, we do remark that we do not need it in its full generality as we will eventually work with a simpler, \emph{symmetric} notion that was alluded to in Section \ref{sec:intro}. Nevertheless, employing the same notation has the advantage that the reader is quite possibly already familiar with it in the context of proving set-multilinear circuit lower bounds.

\begin{definition}[Relative Rank Measure of \cite{LST1,LST2}]
Let $f$ be a polynomial that is set-multilinear with respect to the variable
partition $(X_1, X_2,\ldots, X_d)$ where each set is of size $n$. Let $w = (w_1, w_2,\ldots, w_d)$ be a tuple (or word) of non-zero real numbers such that $2^{|w_i|} \in [n]$ for all $i \in [d]$. For each $i \in [d]$, let $X_i(w)$ be the
variable set obtained by removing arbitrary variables from the set $X_i$ such that $|X_i(w)| = 2^{|w_i|}$, and let $\ol{X}(w)$ denote the tuple of sets of variables $(X_1(w),\ldots,X_d(w))$.
Corresponding to a word $w$, define $\calP_w \coloneqq \{i\ |\ w_i > 0\}$ and $\calN_w \coloneqq \{i\ |\ w_i < 0\}$. Let $\calM^{\calP}_{w}$ be the
set of all set-multilinear monomials over a subset of the variable sets $X_1(w), X_2(w),\ldots, X_d(w)$
indexed by $\calP_w$, and similarly let $\calM^{\calN}_{w}$ be the set of all set-multilinear monomials over these variable
sets indexed by $\calN_w$.

Define the ‘partial derivative matrix’ matrix $\calM_w(f)$ whose rows are indexed by the elements
of $\calM^{\calP}_w$ and columns indexed by the elements of $\calN^{\calP}_w$ as follows: the entry of this matrix corresponding to a
row $m_1$ and a column $m_2$ is the coefficient of the monomial $m_1\cdot m_2$ in $f$. We define
\[
\rk_w(f) \coloneqq \frac{\mathrm{rank}(\calM_w(f))}{\sqrt{|\calM^{\calP}_w|\cdot |\calM^{\calN}_w|}} = \frac{\mathrm{rank}(\calM_w(f))}{2^{\frac{1}{2}\sum_{i\in[d]}|w_i|}}.
\]
\end{definition}

\begin{definition}
For any tuple $w = (w_1,\ldots, w_t)$ and a subset $S \subseteq [t]$, we shall refer to the
sum $\sum_{i\in S} w_i$ by $w_S$. And by $w|_S$, we will refer to the tuple obtained by considering only the
elements of $w$ that are indexed by $S$. We denote by $\Fsm[\calT]$ the set of set-multilinear polynomials over the tuple
of sets of variables $\calT$.
\end{definition}

The following is a simple result that establishes various useful properties of the relative rank measure.
\begin{claim}[\cite{LST1}]\label{clm:rk-props}
\begin{enumerate}
    \item(Imbalance) Say $f \in \Fsm[\ol{X}(w)]$. Then, $\rk_w(f)\leq 2^{-|w_{[d]}|/2}$.
    \item(Sub-additivity) If $f,g \in \Fsm[\ol{X}(w)]$, then $\rk_w(f+g)\leq \rk_w(f)+\rk_w(g)$.
    \item(Multiplicativity) Say $f = f_1 f_2\cdots f_t $ and assume that for each $i\in [t]$, $f_i\in \Fsm[\ol{X}(w|_{S_i})]$, where $(S_1, \ldots, S_t)$ is a partition of $[d]$. Then
    \[
    \rk_w(f) = \prod_{i\in[t]} \rk_{w|_{S_i}}(f_i).
    \]
\end{enumerate}
\end{claim}

\section{Main Result}\label{sec:main}

We are now ready to prove our main result. We start by showing that the \emph{symmetric} relative rank is large for the NW polynomial. 

\begin{claim}\label{clm:NW-full-rk}
For an integer $n = 2^k$ and $d\leq n$, let $w \in \{k,-k\}^d$ with $w_{[d]} = 0$. Then $\rk_w(NW_{n,d}) = 1$ i.e., $\calM_w(NW_{n,d})$ has full rank.
\end{claim}
\begin{proof}
Fix $n = 2^k$ and $d$, so that we can also write $NW$ for $NW_{n,d}$, and let $n' = d/2$. The condition on $w$ implies that $|\calP_w| = |\calN_w| = n'$. Observe that $\calM_w(NW)$ is a square matrix of dimension $|\calM^{\calP}_{w}| = |\calM^{\calN}_{w}| = n^{n'}$. Consider a row of $\calM_w(NW)$ indexed by a monomial $m_1 = x_{i_1,j_1}\cdots x_{i_{n'},j_{n'}}\in \calM^{\calP}_{w}$. $m_1$ can be thought of as a map from $S = \{j_1,\ldots,j_{n'}\}$ to $\F_n$ which sends $j_\ell$ to $i_\ell$ for each $\ell \in [n']$. Next, by interpolating the pairs $(j_1,i_1),\ldots, (j_{n'},i_{n'})$, we know that there exists a unique polynomial $f(z)\in \F_n(z)$ of degree $<n'$ for which $f(j_\ell) = i_\ell$ for each $\ell\in [n']$. As a consequence, there is a unique `extension' of the monomial $x_{i_1,j_1}\cdots x_{i_{n'},j_{n'}}$ that appears as a term in $NW$, which is precisely $m_1\cdot \prod_{j\in \calN_w}x_{f(j),j}$. Therefore,
all but one of the entries in the row corresponding to $m_1$ must be zero, and the remaining entry must be $1$. Applying the same argument to the columns of $\calM_w(NW)$, we deduce that $\calM_w(NW)$ is a permutation matrix, and so has full rank.
\end{proof}

The following is a more precise and general version of Theorem \ref{thm-intro:main-bd-depth} that is stated in Section \ref{sec:intro}. We also incorporate Remark~\ref{rem:true-bd} here and show our lower bound for any degree $d\leq n$. Theorem \ref{thm-intro:main-bd-depth} follows from the special case $d = n$.  

\begin{theorem}\label{thm:main-bd-depth}
For an integer $n = 2^k$,
let $\F_n$ be a field of size $n$. Let $d,\Delta$ be integers such that $d\leq n$ is large enough\footnote{We only need $d$ to be larger than some absolute constant.} and $\Delta \leq \log d/ \log \log d$. Let $X_i$ denote the set of $n$ variables
$\{x_{i,j} : j \in [d]\}$ and $X$ be the tuple $(X_1,\ldots, X_d)$. Then, any set-multilinear
formula family of product-depth $\Delta$ computing $NW_{n,d}(X)$ must have size at least $d^{\Omega(d^{1/\Delta}/\Delta)}$.
\end{theorem}

\begin{proof}
We show that the symmetric relative rank of low-depth set-multilinear formulas is small with high probability in the lemma below, and then combine it with Claim \ref{clm:NW-full-rk} above to prove the desired bound.
\begin{lem}\label{lem:rk-bd-depth}
Let $C$ be a set-multilinear formula  of product-depth $1\leq \Delta \leq \log d/ \log \log d$ of size at most $s$ which computes a polynomial that is set-multilinear with respect to the partition $(X_1,\ldots, X_d)$ where each $|X_i| = n$. Let $w \in \{k,-k\}^d$ be chosen uniformly at random. Then, we have
\[
\rk_w(C)\leq  s \cdot 2^{-\frac{kd^{1/\Delta}}{20}}
\]
with probability at least $1 - s\cdot d^{-\frac{d^{1/\Delta}}{12\Delta}}$.
\end{lem}
\begin{proof}
We prove the statement by induction on $\Delta$.

If $\Delta = 1$, then $C =  C_1 + \cdots +C_t$ where each $C_i$ is a product of linear forms. So, for all $i\in [t]$, by Claim \ref{clm:rk-props},
\[
\rk_w(C_i) = \prod_{i=1}^d 2^{-\frac{1}{2}|w_j|} = 2^{-\frac{kd}{2}}
\]
where in the last step, we used the observation that regardless of the choice of $w$, $|w_j| = k$ for all $j\in [n]$. Hence, by the sub-additivity of $\rk_w$, with probability $1$, we have
\[\rk_w(C) \leq s\cdot 2^{-\frac{kd}{2}}\leq  s\cdot 2^{-\frac{kd}{20}}.
\]

Next, we assume the statement is true for all formulas of product-depth $\leq \Delta$. Let $C$ be a formula of
product-depth $\Delta + 1$. 
So, $C$ is of the form $C = C_1 + \cdots + C_t$. Following an overall proof strategy similar to the one in \cite{LST1}, we say that a sub-formula $C_i$ of size $s_i$ is of type 1 if one of its factors has
degree at least $T_\Delta = d^{\frac{\Delta}{\Delta+1}}$, otherwise we say it is of type 2.

Suppose $C_i = C_{i,1}\cdot \cdots \cdot C_{i,t_i}$ is of type 1 with, say, $C_{i,1}$ having degree at least $T_\Delta$. Let $w^{i,1}$ be the corresponding word i.e., $w^{i,1} = w|_{S_1}$ if $C_{i,1}$ is set-multilinear with respect to $S_1\subsetneq [d]$. If it has size $s_{i,1}$, then since it has product-depth at most $\Delta$, it follows by induction that 
\[
\rk_w(C_i) \leq \rk_{w^{i,1}}(C_{i,1}) \leq s_{i,1}\cdot 2^{-\frac{kT_{\Delta}^{1/\Delta}}{20}} \leq s_{i}\cdot 2^{-\frac{kd^{1/(\Delta+1)}}{20}}
\]
with probability at least 
\[
1- s_{i,1}\cdot T_\Delta^{-\frac{T_\Delta^{1/\Delta}}{12\Delta}} \geq 1- s_{i}\cdot d^{-\frac{d^{1/(\Delta+1)}}{12\Delta }\cdot \frac{\Delta}{\Delta + 1}} = 1- s_{i}\cdot d^{-\frac{d^{1/(\Delta+1)}}{12(\Delta+1)}}.
\]

Now suppose that  $C_i= C_{i,1}\cdot \cdots \cdot C_{i,t_i}$ is of type 2 i.e., each factor $C_{i,j}$ has degree $<T_\Delta$. Note that this forces $t_i> d/T_\Delta = d^{ \frac{1}{\Delta + 1}}$. As the formula is set-multilinear, $(S_1, \ldots, S_{t_i})$ form a partition of $[d]$
where each $C_{i,j}$ is set-multilinear with respect to $(X_\ell)_{\ell\in S_j}$ and $C_i$ is set-multilinear with
respect to $(X_\ell)_{\ell\in S}$. Let $w^{i,1},\ldots, w^{i,t_i}$ be the corresponding decomposition, whose respective sums are denoted simply by $w_{S_1},\ldots,w_{S_{t_i}}$.

From the properties of $\rk_w$ (Claim~\ref{clm:rk-props}), we have
\[
\rk_w(C_i) = \prod_{j=1}^{t_i} \rk_{w^{i,j}}(C_{i,j}) \leq \prod_{j=1}^{t_i} 2^{-\frac12 |w_{S_j}|} = 2^{-\frac12\sum_{j=1}^{t_i}|w_{S_j}|},
\]
from which we observe that the task of upper bounding $\rk_w(C)$ can be reduced to the task of lower bounding the sum $\sum_{j=1}^{t_i}|w_{S_j}|$, which is established in the following claim. For the sake of convenience, the choice of the alphabet for $w$ below is scaled down to $\{-1,1\}$.

\begin{claim}\label{clm:sum-lb}
For large enough $d$, suppose $(S_1, \ldots, S_{\ell})$ is a partition of $[d]$ such that each $|S_j| < T_\Delta = d^{\frac{\Delta}{\Delta+1}}$. Then, we have
\[
\Pr_{w\sim \{-1,1\}^d}\left[\sum_{j=1}^{\ell}|w_{S_j}| < \frac{d^{1/(\Delta+1)}}{10}\right] \leq d^{-\frac{d^{1/(\Delta+1)}}{12}}.
\]
\end{claim}
\begin{proof}
We first show that without loss of generality, we may assume that each $S_j$ has size `roughly' $T_\Delta$. To see this, we apply the following \emph{clubbing} procedure to the sets in the partition $(S_1, \ldots, S_{\ell})$: 
\begin{itemize}
    \item Start with the given partition $(S_1, \ldots, S_{\ell})$. At each step in the procedure, we shall `club' two of the sets in the partition according to the following rule.
    \item If there are two distinct sets $S'$ and $S''$ in the current partition each of size $< T_\Delta/2$, we remove both of them and add their union $S'\cup S''$ to the partition.
    \item If the rule above is no longer applicable, then we have at most one set in the current partition of size $<T_\Delta/2$. If there is none, then we halt the procedure here. Otherwise, we union this set with any one of the other sets  and then halt.
\end{itemize}
After the clubbing procedure, we are left with a partition $(S_1',\ldots,S_{\ell'}')$ of $[d]$ such that $\frac{T_\Delta}{2}\leq |S_j'|\leq \frac{3T_\Delta}{2}$ for each $j\in [\ell']$, also implying that $\frac{2d^{1/(\Delta+1)}}{3}\leq \ell'\leq 2d^{1/(\Delta+1)}$. Through a repeated use of the triangle inequality, we see that $\sum_{j=1}^{\ell'}|w_{S'_j}|\leq \sum_{j=1}^{\ell}|w_{S_j}|$. Therefore, upper bounding the latter sum is a `smaller' event than upper bounding the former sum. Hence, it suffices to prove the statement of the claim with the assumption that $\frac{T_\Delta}{2}\leq |S_j|\leq \frac{3T_\Delta}{2}$ for each $j\in [\ell]$ (we henceforth drop the primed notation).

Now, in the event that the sum $\sum_{j=1}^{\ell}|w_{S_j}|$ is at most $\frac{d^{1/(\Delta+1)}}{10}$, since $\ell\geq \frac{2d^{1/(\Delta+1)}}{3}$, it follows that for at least half of the sets $S_j$, $w_{S_j} = 0$ (as $\frac{2}{3} - \frac{1}{10} = \frac{17}{30}>\frac12$). By Stirling's approximation, it follows that for a fixed $j$, the probability 
\[
\Pr_{w\sim \{-1,1\}^d}\left[w_{S_j} = 0\right]\leq \sqrt{\frac{2}{\pi |S_j|}}\leq \sqrt{\frac{4}{\pi T_\Delta}} = \sqrt{\frac{4}{\pi}}\cdot \frac{1}{d^{\frac{\Delta}{2(\Delta+1)}}}< \frac{2}{d^{1/3}},
\]
where in the final step, we used $\Delta\geq 2$. Therefore, the probability that this happens for $\ell/2$ distinct $j$ is bounded by
\[\binom{\ell}{\ell/2} \cdot \left(\frac{2}{d^{1/3}}\right)^\frac{\ell}{2}< 2^\ell\cdot \left(\frac{2}{d^{1/3}}\right)^\frac{\ell}{2} = \left(\frac{2\sqrt{2}}{d^{1/6}}\right)^\ell \leq \left(\frac{2}{d^{1/9}}\right)^{{d^{1/(\Delta+1)}}}< d^{-\frac{d^{1/(\Delta+1)}}{12}},\]
where we used the bound $\ell\geq \frac{2d^{1/(\Delta+1)}}{3}$.
\end{proof}

The claim above and the preceding calculation immediately implies that for a sub-formula $C_i$ of type 2, 
\[
\rk_w(C_i) \leq s_{i}\cdot 2^{-\frac{kd^{1/(\Delta+1)}}{20}}
\]
with probability at least $1-d^{-\frac{d^{1/(\Delta+1)}}{12}}\geq 1 - s_i\cdot d^{-\frac{d^{1/(\Delta+1)}}{12(\Delta+1)}}$.

Next, by a union bound over $i\in [t]$ and the sub-additivity property of $\rk_w$, it follows that 
\[
\rk_w(C) \leq \rk_w(C_1) +\cdots + \rk_w(C_t) \leq s_1 \cdot 2^{-\frac{kd^{1/(\Delta+1)}}{20}} + \cdots + s_t \cdot 2^{-\frac{kd^{1/(\Delta+1)}}{20}} = s \cdot 2^{-\frac{kd^{1/(\Delta+1)}}{20}}
\]
with probability at least $1 - s\cdot d^{-\frac{d^{1/(\Delta+1)}}{12(\Delta+1)}}$, which concludes the proof of the lemma.
\end{proof}
Returning to the proof of the theorem, let $C$ be a set-multilinear formula of product depth $\Delta$ of size $s$ computing $NW_{n,d}(X)$. Suppose $s < d^{\frac{d^{1/\Delta}}{24\Delta}}$. Then, by Lemma \ref{lem:rk-bd-depth}, with probability at least $1 - d^{-\frac{d^{1/\Delta}}{24\Delta}}$,
\[
\rk_w(C)\leq  s \cdot 2^{-\frac{kd^{1/\Delta}}{20}}.
\] 
But now, we can condition on the event that $w_{[d]} = 0$ (which occurs with probability $\Theta(\frac{1}{\sqrt{d}}$)) to establish the existence of a word $w\in \{-k,k\}^d$ with $w_{[d]} = 0$ such that $w$ satisfies $
\rk_w(C)\leq  s \cdot 2^{-\frac{kd^{1/\Delta}}{20}}$. This is because of the asymptotic bound $\frac{1}{\sqrt{d}} \gg d^{-\frac{d^{1/\Delta}}{24\Delta}}$, which follows from the given constraints on the parameters $d,\Delta$. Therefore, by Claim \ref{clm:NW-full-rk},
\[
s\geq 2^{\frac{kd^{1/\Delta}}{20}}\cdot \rk_w(C) = n^{\frac{d^{1/\Delta}}{20}}
\]
which contradicts the assumption that $s < d^{\frac{d^{1/\Delta}}{24\Delta}}$. Thus, we conclude that $s\geq d^{\frac{d^{1/\Delta}}{24\Delta}} = d^{\Omega(d^{1/\Delta}/\Delta)}$.

\end{proof}

Next, we show the supplementary result (Theorem \ref{thm-intro:main-gen-depth}) mentioned in Section \ref{sec:intro}, stated more precisely below.

\begin{theorem}\label{thm:main-gen-depth}
For an integer $n = 2^k$,
let $\F_n$ be a field of size $n$ and suppose $d\leq n$ is large enough. Let $X_i$ denote the set of $n$ variables
$\{x_{i,j} : j \in [n]\}$ and $X$ be the tuple $(X_1,\ldots, X_d)$.  Then, any set-multilinear
formula family computing  $NW_{n,d}(X)$ must have size at least $d^{\Omega(\log d)}$.
\end{theorem}

\begin{proof}
We first need the following structural result, whose proof can be immediately extrapolated from \cite{Saptarishi-survey} (see Lemma 13.3), where it is shown for multilinear and homogeneous formulas.
\begin{lem}[Product Lemma]\label{lem:prod-ub-depth}
Assume that $F$ is a formula with at most $s$ leaves, and is set-multilinear with respect to the set partition $(X_1,\ldots,X_d)$.
Then, we can write
\[
F = \sum_{i = 1}^s \prod_{j=1}^\ell F_{i,j}
\]
where $\ell \geq \log_3 d$ and for each $i\in [s]$, the product $F_i = \prod_{j=1}^\ell F_{i,j}$ is also set-multilinear. Furthermore, the degrees of $F_{i,j}$ satisfy the following geometric decay property:
\[
\left(\frac{1}{3}\right)^j d \leq \deg(F_{i,j})\leq \left(\frac{2}{3}\right)^j d, \text{ and } \deg(F_{i,\ell}) = 1.
\]
\end{lem}

\begin{lem}\label{lem:rk-gen-depth}
Let $F$ be a set-multilinear formula of size at most $s$ which computes a polynomial that is set-multilinear with respect to the partition $(X_1,\ldots, X_d)$ where each $|X_i| = n$. Let $w \in \{k,-k\}^d$ be chosen uniformly at random. Then, we have
\[
\rk_w(C)\leq  s \cdot 2^{-\frac{k\log d}{20}}
\]
with probability at least $1 - s\cdot d^{-\frac{\log d}{60}}$.
\end{lem}
\begin{proof}
We begin by writing $F$ in the form that is given by Lemma \ref{lem:prod-ub-depth}. Now, because of the geometric decay of the degrees of $F_{i,j}$, we observe that for each $i\in [s]$, at least for the first $\frac{3\ell}{4}$ many values of $j$, $\deg(F_{i,j})\geq d^{1/4}$. In other words, at least a \emph{constant} fraction of the $F_{i,j}$s have their degrees at least \emph{polynomially large} in $d$. This observation will be instrumental in establishing the following claim, which is akin to Claim \ref{clm:sum-lb} used in the proof of Lemma \ref{lem:rk-bd-depth}.

\begin{claim}\label{clm:sum-lb-easy}
For large enough $d$, suppose $(S_1, \ldots, S_{\ell})$ is a partition of $[d]$ such that $\left(\frac{1}{3}\right)^j d \leq |S_j|\leq \left(\frac{2}{3}\right)^j d$ for all $j\in[\ell]$, and  $|S_\ell| = 1$. Then, we have
\[
\Pr_{w\sim \{-1,1\}^d}\left[\sum_{j=1}^{\ell}|w_{S_j}| < \frac{\log d}{10}\right] \leq d^{-\frac{\log d}{60}}.
\]
\end{claim}
\begin{proof}
Consider the given event that $\frac{\log d}{10}$ exceeds the sum $\sum_{j=1}^{\ell}|w_{S_j}|$. As $\ell \geq \frac{\log d}{\log 3} > \frac{5\log d}{8}$, it follows that for at least half of the sets $S_j$, $w_{S_j} = 0$ (since $\frac{5}{8}-\frac{1}{10} = \frac{21}{40}>\frac12$). By the observation above, it also follows that at least for $\frac{\ell}{4}$ many of the \emph{first} $\frac{3\ell}{4}$ values of $j$, $w_{S_j} = 0$. But for a fixed such $j$, since $|S_j|\geq d^{1/4}$, the probability 
\[
\Pr_{w\sim \{-1,1\}^d}\left[w_{S_j} = 0\right]\leq \sqrt{\frac{2}{\pi |S_j|}}< \frac{1}{\sqrt{|S_j|}} \leq \frac{1}{d^{1/8}},
\]
Therefore, the probability that this happens for $\ell/4$ distinct $j$ amongst the first $\frac{3\ell}{4}$ values of $j$ is bounded by
\[\binom{3\ell/4}{\ell/4} \cdot \left(\frac{1}{d^{1/8}}\right)^\frac{\ell}{4}< 2^{3\ell/4}\cdot \left(\frac{1}{d^{1/8}}\right)^\frac{\ell}{4} < \left(\frac{2}{d^{1/32}}\right)^\ell < d^{-\frac{\log d}{60}}.\]
\end{proof}

By sub-additivity of $\rk_w$ (Claim~\ref{clm:rk-props}), we have
\begin{equation}\label{eqn:subadd}
    \rk_w(F)\leq \rk_w(F_1)+\cdots + \rk_w(F_s).
\end{equation}

So, fix an $i\in [s]$. As the formula is set-multilinear, let $(S_1, \ldots, S_{\ell})$ be the partition of $[d]$
such that each $F_{i,j}$ is set-multilinear with respect to $(X_t)_{t\in S_j}$. Let $w^{i,1},\ldots, w^{i,\ell}$ be the corresponding decomposition, whose respective sums are denoted by $w_{S_1},\ldots,w_{S_{\ell}}$. Then, by Claim \ref{clm:sum-lb-easy},
\[
\rk_w(F_i) = \prod_{j=1}^{\ell} \rk_{w^{i,j}}(F_{i,j}) \leq \prod_{j=1}^{\ell} 2^{-\frac12 |w_{S_j}|} = 2^{-\frac12\sum_{j=1}^{\ell}|w_{S_j}|}\leq 2^{-\frac{k\log d}{20}}
\]
with probability at least $1- d^{-\frac{\log d}{60}}$. Therefore, by a union bound over $i\in[s]$ and (\ref{eqn:subadd}), we conclude that \[
\rk_w(F)\leq  s \cdot 2^{-\frac{k\log d}{20}}
\]
with probability at least $1 - s\cdot d^{-\frac{\log d}{60}}$.
\end{proof}

Returning to the proof of the theorem, let $F$ be a set-multilinear formula of size $s$ computing $NW_{n,d}$. Suppose $s < d^{\frac{\log d}{120}}$. Then, by Lemma \ref{lem:rk-gen-depth}, with probability at least $1 - d^{-\frac{\log d}{120}}$,
\[
\rk_w(F)\leq  s \cdot 2^{-\frac{k{\log d}}{20}}.
\] 
But now, we can condition on the event that $w_{[d]} = 0$ (which occurs with probability $\Theta(\frac{1}{\sqrt{d}}$)) to establish the existence of a word $w\in \{-k,k\}^d$ with $w_{[d]} = 0$ such that $w$ satisfies $
\rk_w(F)\leq  s \cdot 2^{-\frac{k{\log d}}{20}}$. This is because of the trivial asymptotic bound $\frac{1}{\sqrt{d}} \gg d^{-\frac{\log d}{120}}$. Therefore, again by Claim \ref{clm:NW-full-rk},
\[
s\geq 2^{\frac{k{\log d}}{20}}\cdot \rk_w(F) = n^{\frac{\log d}{20}}
\]
which contradicts the assumption that $s < d^{\frac{\log d}{120}}$. Thus, we conclude that $s\geq d^{\frac{\log d}{120}} = d^{\Omega(\log d)}$.
\end{proof}

\section{Discussion and Open Problems}\label{sec:open}

We conclude by mentioning some interesting directions for future work.

\begin{itemize}

\item The most interesting and natural question is to make the hard polynomial in our main result IMM$_{n,n}$. This would imply super-polynomial algebraic formula lower bounds.  As far as we know, it is conceivable that our complexity measure could be used to prove the lower bound for the IMM$_{n,n}$ polynomial. While the relative rank of IMM$_{n,n}$ itself is low, there might be a suitable ``restriction" of it such that for a randomly chosen $w\in \{-k,k\}^n$, with reasonably high probability the restriction has large rank. This could then be used to prove the lower bound for IMM$_{n,n}$ (using Lemma \ref{lem:rk-bd-depth} or Lemma \ref{lem:rk-gen-depth}). The result from~\cite{LST1} also showed its lower bound for the IMM polynomial by first analyzing a suitable restriction of IMM (although unfortunately that very same restriction idea does not work for us; please see the discussion in the appendix). Perhaps an intermediate question is to make the hard polynomial computationally simpler, for instance to find any hard polynomial that lies in VP.

\item Another interesting question is to prove an improved depth hierarchy theorem for constant-depth set-multilinear formulas. \cite{LST1} shows a depth hierarchy theorem for low-depth set-multilinear formulas. However, since their lower bounds only hold for small degrees, the depth hierarchy theorem in~\cite{LST1} only gives a quasi-polynomial separation of successive product-depths. It would be very interesting to obtain exponential separations (which for instance have been shown for low-depth multilinear circuits in \cite{ChillaraEL018}) using our measure.

\item Another interesting direction could be to obtain lower bounds for general set-multilinear circuits via improved depth reduction results. 
The work of Kumar, Oliveira, and Saptharishi (\cite{KOS19}) provides some insight in this context, which shows an improved depth reduction to product-depth $\Delta$ with a size blow-up of $N^{O(\Delta\cdot (N/\log N)^{1/\Delta})}$ for {multilinear} circuits (regardless of degree). If a similar improvement (or any asymptotic improvement in the exponent) on the bound for general circuits from \cite{Tavenas15} could be shown to hold for set-multilinear circuits in the setting of Theorem \ref{thm-intro:main-bd-depth} or Theorem \ref{thm:main-bd-depth} (i.e., when $N\geq d^2$), then combined with our lower bounds, this would imply super-polynomial set-multilinear circuit lower bounds. We should note that \cite{FournierLMS15} rules out the possibility of obtaining a stronger reduction to depth-4, or $\Sigma\Pi\Sigma\Pi$ circuits, as it shows an $n^{\Omega(\sqrt{n})}$ size lower bound for set-multilinear depth-4 circuits computing IMM$_{n,n}$, which of course has small polynomial-sized set-multilinear circuits. Nevertheless, there is still the possibility of obtaining improved depth reduction statements for product-depths 2 (which as noted earlier, is $\Sigma\Pi\Sigma\Pi\Sigma$ and hence more general than depth-4) or higher, and combining it with our Theorem \ref{thm-intro:main-bd-depth} to obtain unbounded-depth set-multilinear circuit lower bounds. \cite{KumarS16} shows a quasi-polynomial separation between the strength of homogeneous $\Sigma\Pi\Sigma\Pi$ and $\Sigma\Pi\Sigma\Pi\Sigma$ circuits, which could be considered as some evidence towards the validity of this possibility.

\end{itemize}

\section*{Acknowledgments}
We would like to thank Swastik Kopparty, Mrinal Kumar, and Ben Rossman for several helpful discussions.

\bibliographystyle{alpha}
\bibliography{bibfile}

\newpage 

\appendix

\section{Word Polynomials from \cite{LST1,LST2} and Our Measure}

Both \cite{LST1,LST2} show their set-multilinear formula lower bounds for IMM$_{n,d}$ by showing that small enough set-multilinear formulas have low relative rank and that a certain ``restriction'' of IMM$_{n,d}$ has large relative rank. This restriction possesses the desirable property that if there is a small low-depth set-multilinear circuit computing IMM$_{n,d}$, then there is one for this restriction as well. It is then natural to wonder if we can use these same restrictions for our \emph{symmetric} measure and deduce strong lower bounds for IMM (in order to show super-polynomial general formula lower bounds as discussed), in addition to obtaining them for the NW polynomial. Unfortunately, it is straightforward to show that this is not possible, as we shall now see.

\begin{definition}[Word Polynomials of \cite{LST1,LST2}]
Let $w\in \R^d$ be any word with non-zero entries. Say $X(w) = (X_1,\ldots, X_d)$ where each $X_i$ has size $2^{|w_i|}$; we assume that the variables of $X_i$ are labelled
by strings in $\{0, 1\}^{|w_i|}$.

Given any monomial $m\in \Fsm[\ol{X}(w)]$, let $m_+$ denote the corresponding ``positive" monomial
from $\calM^\calP_w$ and $m_-$ the corresponding ``negative" monomial from $\calM^\calN_w$. As each variable of $\ol{X}(w)$
is labelled by a Boolean string and each set-multilinear monomial over any subset of $\ol{X}(w)$ is
associated with a string of variables, we can associate any such monomial $m'$ with a Boolean
string $\sigma(m')$. More precisely, if $j_1<\cdots<j_t$ and $m'= x_{\sigma_1}^{(j_1)}x_{\sigma_1}^{(j_1)} \ldots x_{\sigma_t}^{(j_t)}$ with $x_{\sigma_i}^{(j_i)}\in X_{j_i}$
and $\sigma_i\in \{0,1\}^{|w_{j_i}|}$ for each $i\in [t]$, then $\sigma(m')$ is defined to be $\sigma_1\cdots\sigma_t$. 
We will write $\sigma(m_+)\sim \sigma(m_-)$ when the shorter one is a prefix of the other one. The polynomial $P_w$ is defined as follows
\[
P_w  = \sum_{\substack{m\in \F[\ol{X}(w)],\\ \sigma(m_+)\sim \sigma(m_-)}} m.
\]
\end{definition}

Clearly, the matrices $\calM_w(P_w)$ are full-rank (i.e., have rank equal to either the number of rows
or the number of columns, whichever is smaller). So, $\rk_w(P_w) = 2^{-|w_{[d]}|/2}$.

In our measure, $w\in \{k,-k\}^d$ with $w_{[d]} = 0$ i.e., there is an equal number of positive and negative variable sets and an equal number of variables $n = 2^k$ in each set. Thus, in the sum above, $\sigma(m_+)\sim \sigma(m_-)$ gets replaced with $\sigma(m_+)= \sigma(m_-)$. The sum is indexed over all Boolean strings of length $kd/2$, and so there are $n^{d/2}$ terms in all. Moreover, there is a canonical bijection between the positive and negative variables: since $|\calP_w| = |\calN_w| = d/2$, if an element $j\in \calP_w$ is the $k$-th largest element in $\calP_w$, it corresponds to the $k$-th largest element $j'$ in $\calN_w$ such that $x_{i,j}$ appears in a monomial of $P_w$ if and only if so does $x_{i,j'}$. Let $\phi:\calP_w \rightarrow \calN_w$ denote this correspondence. Then, we see that 
\[
P_w = \prod_{j\in\calP_w} \sum_{i=1}^n x_{i,j}\cdot x_{i,\phi(j)},
\]
implying that $P_w$ actually has small depth-3 set-multilinear formulas.
\end{document}